\documentclass[11 pt, onecolumn]{article}  

\usepackage[latin1]{inputenc}
\usepackage{graphicx}
\usepackage{amssymb}
\usepackage{amsmath}
\usepackage{amsthm}
\usepackage{xspace} 
\usepackage{color}
\usepackage{multirow,tabularx}
\usepackage{algorithm}
\usepackage{algorithmicx}
\usepackage{algpseudocode}
\usepackage{cases}
\usepackage{balance}
\usepackage{psfrag}

\usepackage{enumitem}
\usepackage[normalem]{ulem}
\usepackage{yhmath}
\usepackage{a4wide}

\newtheorem{lemma}{Lemma}
\newtheorem{theorem}{Theorem}

\newtheorem{proposition}{Proposition}
\newtheorem{remark}{Remark}
\newtheorem{definition}{Definition}
\newtheorem{example}{Example}

\def\rr{\mathbb{R}}

\def\cM{\mathcal{M}}

\def\cB{\mathcal{B}}

\def\cL{\mathcal{L}}

\def\cW{\mathcal{W}}

\def\cA{\mathcal{A}}
\def\cD{\mathcal{D}}

\def\cP{\mathcal{P}}

\def\cC{\mathcal{C}}

\renewcommand\labelenumi{\roman{enumi})}
\renewcommand\theenumi\labelenumi

\newcommand{\tg}{\texttt{2P1EG}\xspace}
\newcommand{\mpg}{\texttt{MP1EG}\xspace}

\newcommand{\hull}[1]{\mathbb{H}\left\{#1\right\}}

\def\eps{\varepsilon}

\def\rc{r}             
\def\creg{\mathcal{D}}      
\def\cregm{\mathcal{M}}     
\newcommand{\op}[1]{\mathtt{int}\left\{#1\right\}}  
\newcommand{\sign}[1]{\text{sign}(#1)}           

\newcommand{\ds}{\displaystyle}

\title{\vspace*{-4mm}A Family of Switching Pursuit Strategies\\for a Multi-Pursuer Single-Evader Game}

\author{Marco Casini, Andrea Garulli\thanks{M.~Casini and A.~Garulli are with the Dipartimento di
		Ingegneria dell'Informazione e Scienze Matematiche, Universit\`a di Siena, via Roma~56, 53100 Siena, Italy. E-mail: marco.casini@unisi.it,~andrea.garulli@unisi.it.}}

\date{}

\begin{document}
\maketitle \thispagestyle{empty}
	
\begin{abstract}
This paper introduces a new family of pursuit strategies for multi-pursuer single-evader games in a planar environment. They leverage conditions under which the minimum-time solution of the game becomes equivalent to that of a suitable two-pursuer single-evader game.
This enables the design of strategies in which the pursuers first aim to meet such conditions, and then transition to a two-pursuer game once they are satisfied.
As a consequence, naive strategies that are in general unsuccessful, can be turned into winning strategies by switching to the appropriate two-pursuer game. Moreover, it is shown via numerical simulations that the switching mechanism significantly enhances the performance of existing pursuit algorithms, like those based on Voronoi partitions. 
\end{abstract}

\noindent Keywords: Pursuit-evasion games; multi-agent systems; autonomous agents.

\section{Introduction}\label{sec:intro}

Pursuit-evasion games are intensively studied because they allow to tackle a variety of problems involving multi-agent systems.
The large number of applications treated in the literature range from mobile robotics \cite{chung2011,noori2016} to optimal robot control \cite{macias2018auto}, space operations \cite{venigalla2021jgcd,shi2023aa}, predator-prey dynamics in biological systems \cite{lin2023cdc}, and beyond. Recent years have witnessed a surge in interest towards games involving multiple pursuers \cite{zhou2016,chen2016,kothari2017}, multiple evaders \cite{scott2018auto,bajaj2021cooperative}, or both \cite{lopez2019tac,garcia2020multiple}, leading to a broad spectrum of problem formulations and solutions.

Traditionally, pursuit-evasion problems are approached by framing them as differential games \cite{isaacs1965,Basar1982}.
Unfortunately, an optimal solution is available only for a limited number of games.
An example is that of the minimum-time solution of the two-pursuer one-evader game in a two-dimensional environment \cite{garcia2017,pachter2019,Casini2022jirs}.
For the three-pursuer one-evader game, only suboptimal solutions have been derived \cite{pshenichnyi1976}, although it has been established that optimal strategies must involve switching between different linear sub-paths \cite{ibragimov2016ijpam,Casini2022scl}.
An open-loop minimum-time solution for games with more than three pursuers has been proposed in \cite{von_moll2019}, showing that the agents must move along linear paths. However, a closed-loop solution to this differential game remains elusive.
The above observations motivate the investigation of conditions under which the optimal solution of the two-pursuer game can be exploited to devise successful strategies for games involving more than two pursers.

This paper proposes a new family of pursuit strategies for a multi-pursuer single-evader game, in which agents move in a planar environment with simple motion and equal maximum speed. The core idea is to derive and leverage conditions under which the minimum-time game boils down to a two-pursuer game. This allows one to design pursuit strategies aiming at enforcing these conditions and transitioning to a two-pursuer game when they are satisfied. Although these switching pursuit strategies are in general suboptimal, they consistently reduce capture time compared to strategies without switching. As a result, they enhance the performance of existing pursuit techniques and enable the development of new effective multi-pursuer strategies.

The main contributions of the paper are as follows. First, a solution of the game of kind is derived, i.e., a thorough characterization of all agent configurations in which capture of the evader can be achieved. Then, conditions under which the minimum-time pursuit strategy simplifies to that of a two-pursuer single-evader game are established. The third contribution is the definition of new families of switching pursuit strategies, based on the previously derived conditions, that guarantee capture of the evader. Specifically, simple strategies like pure pursuit or fixed-point pursuit, which in general fail to capture the evader, are transformed into successful strategies through a switch to the appropriate two-pursuer game. Furthermore, it is demonstrated through numerical simulations that switching can substantially reduce the capture time, even for strategies that inherently guarantee capture, such as the Voronoi partition-based approach introduced in \cite{zhou2016}.
On the whole, the proposed switching mechanism provides a powerful framework for the design of a variety of new pursuit strategies in multi-pursuer single-evader games.
A preliminary version of this work has been presented in \cite{casini2024cdc}.

As for the paper organization, the formulation of the multi-pursuer single-evader game is given in Section \ref{sec:formulation}, along with a summary of the  two-pursuer game optimal solution. Section \ref{sec:2P1E_multiP} provides the solution to the game of kind and the conditions under which the minimum-time multi-pursuer game boils down to a two-pursuer one. This leads to the introduction of new switching pursuit strategies in Section \ref{sec:strategies}, whose performance is assessed in Section \ref{sec:simulations} by means of numerical simulations. Concluding remarks and future developments are reported in Section~\ref{sec:conclusions}.

\subsection{Notation and definitions}\label{sec:notation}
Given a vector $V$, its transpose is denoted by $V'$, while $\|V\|$ is its Euclidean norm. For $V,W\in\rr^2$, we denote by $\overline{VW}$ the segment with $V$ and $W$ as endpoints. 
$\cC(P,r)=\{Q\in\rr^2\colon \|Q-P\|\le r\}$ denotes a circle centered in $P\in\rr^2$ with radius $r$. 
For a closed set $\cA$, $\partial\cA$ is the boundary of $\cA$, while $\hull{\cA}$ denotes the convex hull of $\cA$. Given two sets $\cA,\cB$, $\cA\backslash\cB=\{Q\colon Q\in\cA,~Q\notin\cB\}$. The interior of $\cA$ is denoted by $\op{\cA}=\cA\backslash\partial\cA$.

\section{Multi-pursuer single-evader game}\label{sec:formulation}
A \emph{multi-pursuer single-evader game} (hereafter referred to as \mpg) is a game in which $p$ pursuers aim at capturing one evader. The positions of the pursuers and the evader at time $t$ are denoted by $P_k(t)\in\rr^2,k=1,\ldots p$, and $E(t)\in\rr^2$, respectively. The players move in the plane in simple motion, that is
\begin{equation*}\label{eq:ct}
	\left\{
	\begin{array}{l}
		\dot{E}(t)=v_E(t)\,,\\
		\dot{P}_k(t)=v_{P_k}(t)\,,~k=1,\ldots,p
	\end{array}
	\right.
\end{equation*}
where $v_E(t),v_{P_k}(t)\in\rr^2$ are the velocity vectors of the players. 
We assume that all agents have the same maximum velocity, $\|v_E(t)\|\le v_{max}$, $\|v_{P_k}(t)\|\le v_{max}$, $k=1,\ldots,p$, $\forall t\ge0$. Without loss of generality, it is set $v_{max}=1$ and the initial game time at $t=0$.
At each time $t$, $v_E(t)$ and $v_{P_k}(t)$ are computed as functions of the state of the game $\xi(t)=[E'(t)~P_1'(t)~\dots~P_p'(t)]'$. Notice that the agents do not have information on the velocities of the other players. 
The \emph{strategies} of the agents are functions mapping the current state into their chosen velocity vector, $v_E=\phi(\xi)$, $v_{P_k}=\psi_k(\xi)$. Dependence on time is omitted when it is clear from the context.

The aim of the pursuers is to capture the evader, while the evader tries to avoid capture or to protract the game as long as possible. Capture occurs when the distance between at least one pursuer and the evader is equal to the radius of capture $\rc>0$. Hereafter, only initial conditions such that 
\begin{equation}\label{eq:no_terminate}
	\|P_k(0)-E(0)\|>\rc,\quad k=1,\ldots,p
\end{equation}
will be considered. 
A \emph{winning pursuit strategy} is a strategy adopted by the pursuers which guarantees that capture occurs in finite time whatever is the strategy adopted by the evader. The capture time $T$ is such that
\begin{equation}\label{eq:capturecond}
	\|P_k(T)-E(T)\|=\rc, \text{~for some~} k\in\{1,\ldots,p\}.
\end{equation}
A \emph{minimum time pursuit strategy} is a strategy $\psi=[ \psi_1,\dots,\psi_p]$ which is a solution of the min-max problem
\begin{equation}\label{eq:mintime}
	\min_{\psi} \max_\phi~ T.
\end{equation}
For brevity, we will refer to functions $\psi$ and $\phi$ solving problem \eqref{eq:mintime} as \emph{optimal strategies}.
To the best of our knowledge, \eqref{eq:mintime} is still an open problem for $p \geq 3$. Conversely, for $p=2$ the problem has been fully solved \cite{garcia2017,pachter2019,Casini2022jirs}. 
In the following, the solution of the \emph{two-pursuer one-evader game} (\tg) given in \cite{Casini2022jirs} is recalled. It will play a key role in the definition of the class of winning strategies for the \mpg proposed in this work.

\subsection{Minimum-time solution of the two-pursuer one-evader game}\label{subsec:recall_2P1E}

Consider two pursuers $P_1,P_2$ trying to capture one evader. 
Let us define the region
\begin{equation}\label{eq:capture_region}
	\creg_{12}=\op{\hull{\cC(P_1,\rc)\cup\cC(P_2,\rc)}}.
\end{equation}
First, a solution to the game of kind is provided \cite{Casini2022jirs}.
\begin{proposition}\label{prop:2v1kind}
	There exists a winning pursuit strategy for the \tg if and only if $E \in \creg_{12}$.
\end{proposition}

Hence, set $\creg_{12}$ is the \emph{2-pursuer capture region} associated to $P_1$ and $P_2$; if the evader belongs to this region it will be captured by the pursuers in finite time.
An example is depicted in Fig.~\ref{fig:capture_region}. 
To simplify the treatment, w.l.o.g. we will refer to the reference frame reported in Fig.~\ref{fig:capture_region}, where the origin is the midpoint of the pursuers, and they lie on the $x$-axis, i.e., $P_1=[-d,0]'$, $P_2=[d,0]'$, $E=[x,y]'$. 
Let $E\in\creg_{12}\backslash(\cC(P_1,\rc)\cup\cC(P_2,\rc))$ and define
\begin{equation}\label{eq:function_fp}
	T_{12}=\ds\frac{\kappa \rc+|y|\sqrt{\kappa^2-4x^2(\rc^2-y^2)}	}{2(\rc^2-y^2)}
\end{equation}
where $\kappa=d^2-x^2+y^2-\rc^2$. 
Now, let us define the point $H_{12}$ satisfying
\begin{equation}\label{eq:point_Hp}
	\|E-H_{12}\|=\|P_1-H_{12}\|-\rc=\|P_2-H_{12}\|-\rc=T_{12}.
\end{equation}
Within the considered reference frame, one has $H_{12}=[0,\,h_{12}]'$ where
\begin{equation}\label{eq:solution_2vs1}
	h_{12}=
	\begin{cases}
		y+\sign{y}\sqrt{T_{12}^2-x^2} & \text{~~~if } y\neq 0\\[2mm]
		\pm\sqrt{T_{12}^2-x^2} & \text{~~~if } y=0.
	\end{cases}	
\end{equation}
The following result provides a solution to the minimum-time \tg \cite{Casini2022jirs}.
\begin{proposition}\label{prop:opt_strategies_2vs1}
	Let $p=2$. The strategies solving problem \eqref{eq:mintime} require each player to travel at maximum speed along a straight line to the point $H_{12}$ defined by \eqref{eq:point_Hp}. Then, the evader will be captured in $H_{12}$, at time $T_{12}$ given by \eqref{eq:function_fp}.
\end{proposition}

\begin{figure}[t]
	\centering
	\includegraphics[width=0.85\columnwidth]{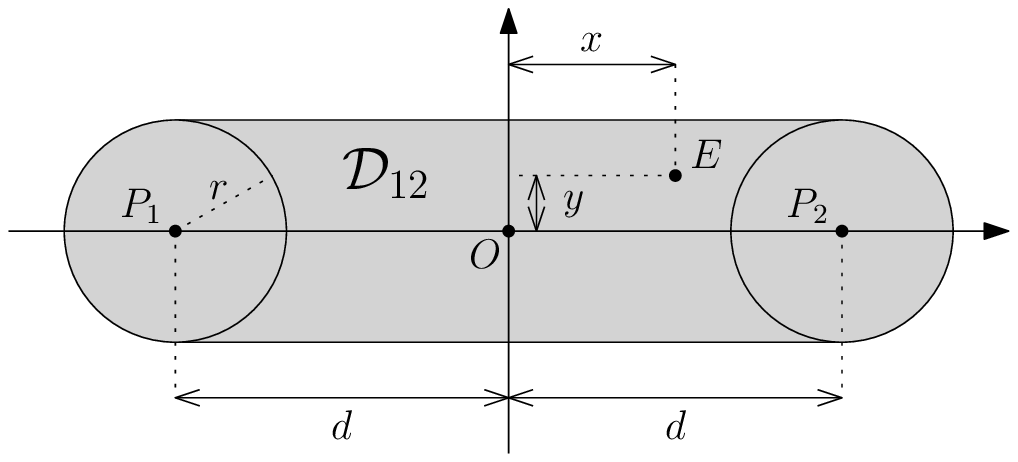}
	\caption{The 2-pursuer capture region for the \tg.}
	\label{fig:capture_region}
\end{figure}

\section{Game of kind and switching conditions}\label{sec:2P1E_multiP}

In this section, the solution of the game of kind for the \mpg is presented. Then, a condition is derived under which the \mpg boils down to a \tg.

For a pair of pursuers $P_i,P_j$, let $\creg_{ij}$ be the corresponding 2-pursuer capture region, given by \eqref{eq:capture_region}. If $E\in\creg_{ij}$, let us denote by $\psi^{ij}$ the minimum-time pursuit strategy for the \tg played by $P_i$ and $P_j$, and by $\phi^{ij}$ the corresponding optimal evader strategy.
Moreover, we denote by $T_{ij}$ the time needed by $P_i$ and $P_j$ to capture $E$ when the agents $P_i$, $P_j$ and $E$ play the optimal \tg strategies.
The capture time $T_{ij}$ and capture point $H_{ij}$ are computed according to \eqref{eq:function_fp}-\eqref{eq:point_Hp}, for each pair $i,j$ such that $E \in \cD_{ij}$.

Let us define the convex hull of the pursuer locations
\begin{equation}\label{eq:hull_P}
	\cP=\hull{P_1,P_2,\ldots,P_p}
\end{equation}
and introduce the \emph{multi-pursuer capture region}, defined as
\begin{equation}\label{eq:capture_region_multi}
	\cregm=\op{\hull{\bigcup_{i=1}^{p} \cC(P_i,\rc)}}.
\end{equation}	
The following result provides the solution of the game of kind for the \mpg.

\begin{theorem}\label{th:game_of_kind}
	If $E\in\cregm$ in \eqref{eq:capture_region_multi}, then there exists a winning pursuit strategy. On the contrary, if $E\notin\cregm$, the evader can avoid capture indefinitely.
\end{theorem}
\begin{proof}
	Let us prove that capture occurs if $E\in\cregm$. It is known that if the evader lies inside the pursuer convex hull $\cP$, there exist pursuit strategies which guarantee capture in finite time, as for instance the Voronoi-based strategy proposed in \cite{zhou2016}. So, it remains to prove that the evader can be captured also when $E\in\cregm\backslash\cP$. It is easy to note that if $E\in\cregm\backslash\cP$ there exist at least two pursuers, say $P_i,P_j$, such that $E\in\creg_{ij}$. Then, the evader will be captured if $P_i,P_j$ play $\psi^{ij}$.\\
	Let us now consider the case $E\notin\cregm$. Since $\cM$ is convex there exists a hyperplane which separates $E$ from $\cregm$. Thus, the evader can escape by going straight orthogonally to such hyperplane, in the opposite direction with respect to $\cregm$ (an example is shown in Fig.~\ref{fig:MP_capture_region}).
	In fact, if $P_i,P_j$ is the pair of pursuers which is closest to the separating hyperplane, the evader will be able to escape from them according to Proposition~\ref{prop:2v1kind}. The other pursuers will neither be able to capture the evader, being further from it than $P_i,P_j$.
\end{proof}

\begin{figure}[h]
	\centering
	\includegraphics[width=0.55\columnwidth]{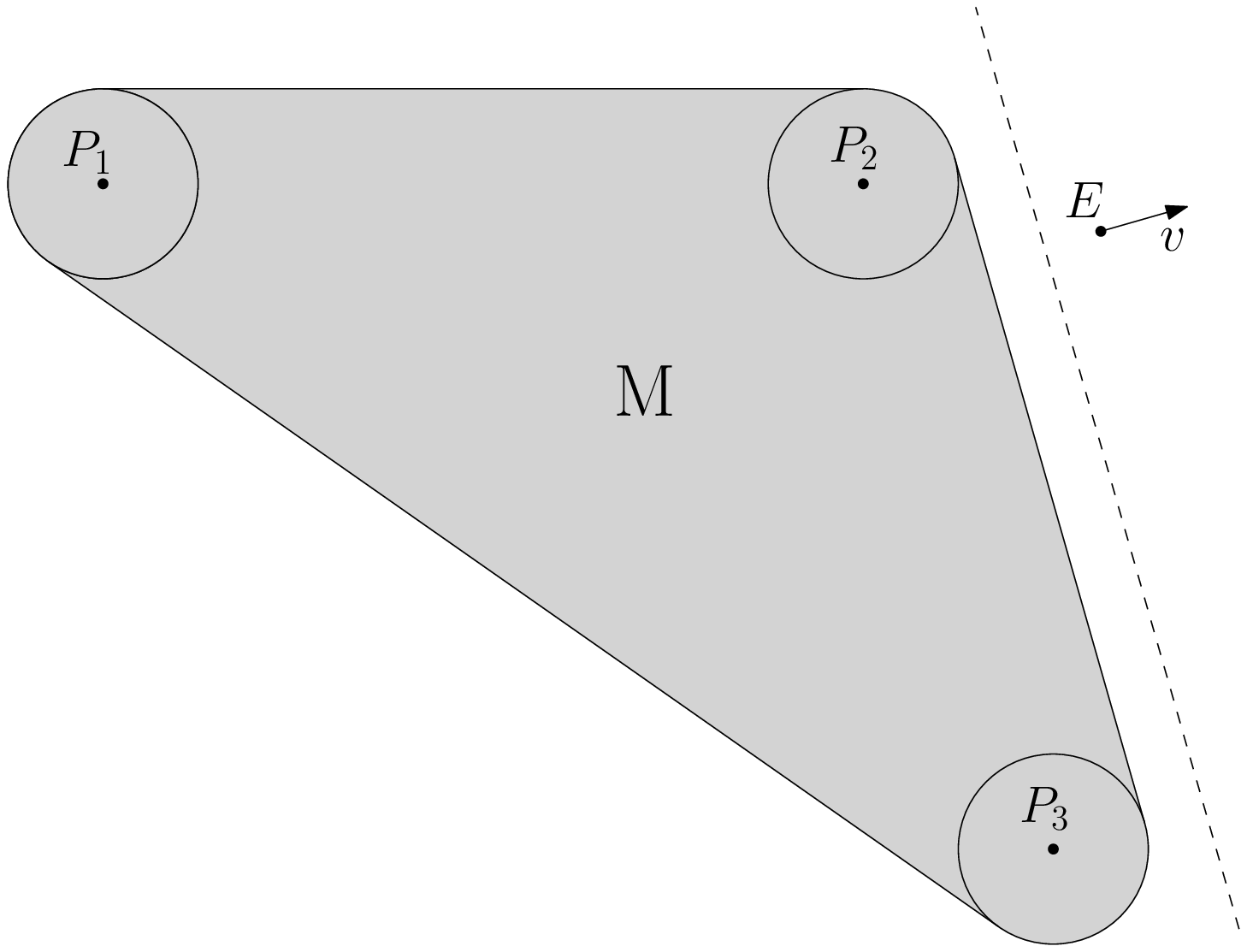}
	\caption{Example of multi-pursuer capture region $\cregm$ for $p=3$ (gray). Since $E\notin\cregm$, the evader can avoid capture by moving along the direction $v$.}
	\label{fig:MP_capture_region}
\end{figure}

The next theorem reports conditions under which the \mpg reduces to the \tg, in the sense that only two pursuers can provide capture in minimum time, irrespectively of the strategy played by the other pursuers.

\begin{theorem}\label{th:P_far1}
	Assume that at time $t$, $E\in\creg_{ij}$ for some $i,j\in\{1,2,\ldots,p\}$, and
	\begin{equation}\label{eq:cond2vs1}
		\|H_{ij}-P_k\|\ge T_{ij}+\rc,\quad k=1,\ldots,p.
	\end{equation}
	Then, from time $t$ onwards, the optimal \mpg strategies for $P_i,P_j,E$ are $\psi^{ij}$ and $\phi^{ij}$, respectively, and capture occurs at time $t+T_{ij}$. The strategies adopted by the other pursuers are irrelevant to the game duration.
\end{theorem}
\begin{proof}
	W.l.o.g. assume $t=0$. If $P_i,P_j$ and $E$ play their optimal strategies for the \tg, capture will occur at the location $H_{ij}$ at time $T_{ij}$. So, $\|H_{ij}-P_i\|=\|H_{ij}-P_j\|=T_{ij}+\rc$.
	Let $P_k$ denote a generic pursuer not involved in the considered \tg, i.e., $k\neq i,j$.	
	We prove that, if condition \eqref{eq:cond2vs1} holds, any strategy adopted by $P_k$ cannot reduce the capture time below $T_{ij}$. 	
	This is equivalent to show that $P_k$ cannot capture the evader during its path to $H_{ij}$.	
	According to Proposition~\ref{prop:opt_strategies_2vs1}, let $v_E=\frac{H_{ij}-E(0)}{\|H_{ij}-E(0)\|}$ denote the velocity vector of the evader. Since $E$ travels towards $H_{ij}$ following a straight path, one has $E(\tau)=E(0)+\tau v_E$, $0\le\tau\le T_{ij}$.	
	Since $\|H_{ij}-E(0)\|=T_{ij}$, then $\cC(E(0),\rc)\subset\cC(H_{ij},T_{ij}+\rc)$ and $\cC(E(\tau),\rc)\subset\cC(H_{ij},T_{ij}+\rc-\tau),$ for all $\tau$ such that $0\le\tau\le T_{ij}$ (see Fig.~\ref{fig:p_farther}).\\
	Consider the case $\|H_{ij}-P_k(0)\|> T_{ij}+\rc$, where strict inequality is assumed in \eqref{eq:cond2vs1}.	
	It follows that $P_k(0)\notin \cC(H_{ij},T_{ij}+\rc)$ and then for any possible trajectory of $P_k$ one has $P_k(\tau)\notin \cC(H_{ij},T_{ij}+\rc-\tau), ~0\le\tau\le T_{ij}$. 	
	Hence, $P_k(\tau)\notin\cC(E(\tau),\rc)$, for all $0\le\tau\le T_{ij}$ and then $P_k$ cannot capture $E$ before $T_{ij}$.\\
	It remains to consider the case when condition \eqref{eq:cond2vs1} holds with equality, i.e., 	
	$\|H_{ij}-P_k\|= T_{ij}+\rc$. Let $Q=E(0)-\rc v_E$. To avoid capture at time $0$, it must be $P_k(0)\neq Q$.
	By going straight towards $H_{ij}$, the $k$-th pursuer will always lie on the border of $\cC(H_{ij},T_{ij}+\rc-\tau)$, thus leading again to $P_k(\tau)\notin\cC(E(\tau),\rc)$ for $0\le\tau<T_{ij}$. When $\tau=T_{ij}$ one will have $\|P_k(T_{ij})-H_{ij}\|=\rc$, meaning that $P_i$, $P_j$ and $P_k$ will capture the evader simultaneously at time $T_{ij}$.
	\begin{figure}[ht]
		\centering
		\includegraphics[width=0.8\columnwidth]{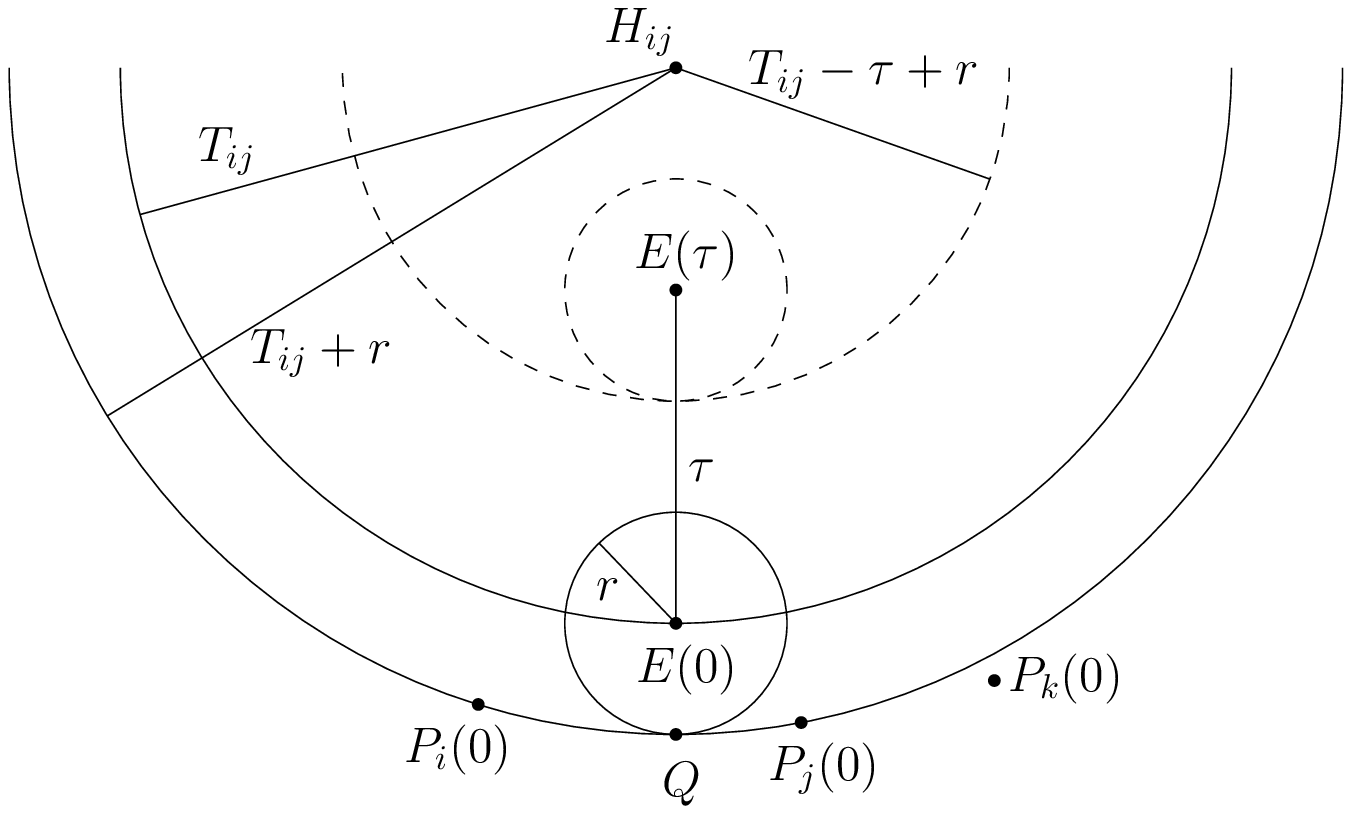}
		\caption{Sketch of the proof of Theorem~\ref{th:P_far1}.}
		\label{fig:p_farther}
	\end{figure}
\end{proof}

Notice that, if $P_i$, $P_j$ and $E$ play their optimal \tg strategies, one has $\|H_{ij}-P_i\|=\|H_{ij}-P_j\|=T_{ij}+\rc$ and so \eqref{eq:cond2vs1} can be rewritten as
\begin{equation}\label{eq:cond2vs1b}
	\|H_{ij}-P_k\|\ge\|H_{ij}-P_i\|,\quad k=1,\ldots,p.
\end{equation}
Such a condition states that if there is no pursuer closer to $H_{ij}$ than $P_i$ (and $P_j$), then \eqref{eq:cond2vs1} holds and the optimal \mpg strategies for $P_i,P_j,E$ are $\psi^{ij}$ and $\phi^{ij}$. This means that the other pursuers cannot improve the capture time even if they go straight to $H_{ij}$. As a direct consequence of Theorem~\ref{th:P_far1}, if there exist two distinct pairs of pursuers $(P_i,P_j)$ and $(P_h,P_l)$ satisfying the conditions of Theorem~\ref{th:P_far1}, then $T_{ij}=T_{hl}$.

Theorem~\ref{th:P_far1} states that if condition \eqref{eq:cond2vs1} or \eqref{eq:cond2vs1b} is satisfied by a pair of pursuers $P_i,P_j$ at a certain time $t$, the strategies of these two pursuers must switch to $\psi^{ij}$ in order to be optimal from time $t$ onwards, irrespectively of what the other pursuers will do. Similarly, the evader must switch to strategy $\phi^{ij}$ in order to maximize its survival time. 
This fact will be exploited in Section \ref{sec:strategies} to devise new families of winning pursuit strategies for the \mpg.
The next result provides a further insight on the way pursuers cooperate to capture the evader after condition \eqref{eq:cond2vs1} has occurred.
\begin{theorem}\label{th:centroid}
	Assume that condition \eqref{eq:cond2vs1} occurs at time $t>0$, and let $P_i,P_j$ denote the pursuers related to the corresponding \tg. If $E(t)\notin \overline{P_i(t) P_j(t)}$ then $\exists k\neq i,j$ such that $H_{ij}(t)$ is the centroid of $P_i(t),P_j(t),P_k(t)$.
\end{theorem}
\begin{proof}
	Since $E(t)\notin \overline{P_i(t) P_j(t)}$, $H_{ij}(t)$ is uniquely defined due to \eqref{eq:solution_2vs1}. By \eqref{eq:cond2vs1b} one has
	\begin{equation}\label{eq:ineq_th3}
		\begin{aligned}
			\|H_{ij}(t)-P_k(t)\|&\ge\|H_{ij}(t)-P_i(t)\|
			=\|H_{ij}(t)-P_j(t)\|, \forall k\neq i,j.
		\end{aligned}
	\end{equation}
	Let $\eps>0$ and set $\tau=t-\eps$. Being $\creg_{ij}$ an open set (see \eqref{eq:capture_region}), it is possible to choose $\eps$ sufficiently small so that $E(\tau)\in\creg_{ij}(\tau)$, with $H_{ij}(\tau)$ defined according to \eqref{eq:point_Hp}. On the other hand, since at time $\tau$ condition \eqref{eq:cond2vs1} does not hold, one has that there exists $k\neq i,j$ such that
	\begin{equation}\label{eq:ineq_th3_bis}
		\|H_{ij}(\tau)-P_k(\tau)\|<\|H_{ij}(\tau)-P_i(\tau)\|=\|H_{ij}(\tau)-P_j(\tau)\|.
	\end{equation}
	By the arbitrariness of $\eps$ and the continuity of the agents' trajectories, \eqref{eq:ineq_th3} and \eqref{eq:ineq_th3_bis} imply that
	\begin{equation*}\label{eq:capture_by_3}
		\|H_{ij}(t)-P_k(t)\|=\|H_{ij}(t)-P_i(t)\|=\|H_{ij}(t)-P_j(t)\|
	\end{equation*}
	and hence $H_{ij}(t)$ is the centroid of $P_i(t),P_j(t),P_k(t)$.
\end{proof}

As a consequence of Theorem~\ref{th:centroid}, if the evader is not aligned with pursuers $P_i$ and $P_j$ when condition \eqref{eq:cond2vs1} is verified, there are at least three pursuers that will end up capturing the evader simultaneously by heading towards their centroid $H_{ij}(t)$.

\section{Switching pursuit strategies}\label{sec:strategies}

In this section, the concept of \emph{switching pursuit strategy} is introduced. The aim of such strategies is to guarantee that condition \eqref{eq:cond2vs1} holds at some finite time, and hence capture is assured by switching to the appropriate \tg strategies. 

\begin{definition}
	A pursuit strategy is a \emph{potential switching strategy} if there exists a pair of pursuers $P_i,P_j$ for which condition~\eqref{eq:cond2vs1} holds at some finite time, for any possible strategy of the evader. A potential switching strategy is referred to as \emph{switching strategy} if $P_i,P_j$ play the optimal \tg strategy as soon as condition~\eqref{eq:cond2vs1} holds.
\end{definition}

By the previous definition, a switching strategy leads to the evader capture in finite time, and so it is a winning strategy. In fact, from the switching time onwards the pursuers will play the optimal \tg strategy, guaranteeing the evader capture.

In general, since condition \eqref{eq:cond2vs1} depends on the evader position, it is not easy to design pursuit strategies that eventually lead to its satisfaction.
Therefore, it is useful to derive sufficient conditions under which \eqref{eq:cond2vs1} holds, that the pursuers can enforce whatever is the strategy played by the evader. Hereafter, two such conditions are presented.

Let us define the union of all the 2-pursuer capture regions as
\begin{equation}\label{eq:Dhat}
	\widehat\creg=\bigcup_{i,j\in\{1,\ldots,p\}}\creg_{ij}.
\end{equation}
An example of set $\widehat\creg$ is shown in Fig.~\ref{fig:union_capture_region}-(a), while the corresponding set $\widehat\creg\backslash\op{\cP}$ is depicted in Fig.~\ref{fig:union_capture_region}-(b). The next theorem provides a sufficient condition for switching to \tg. 

\begin{figure}[t]
	\centering
	\includegraphics[width=0.48\columnwidth]{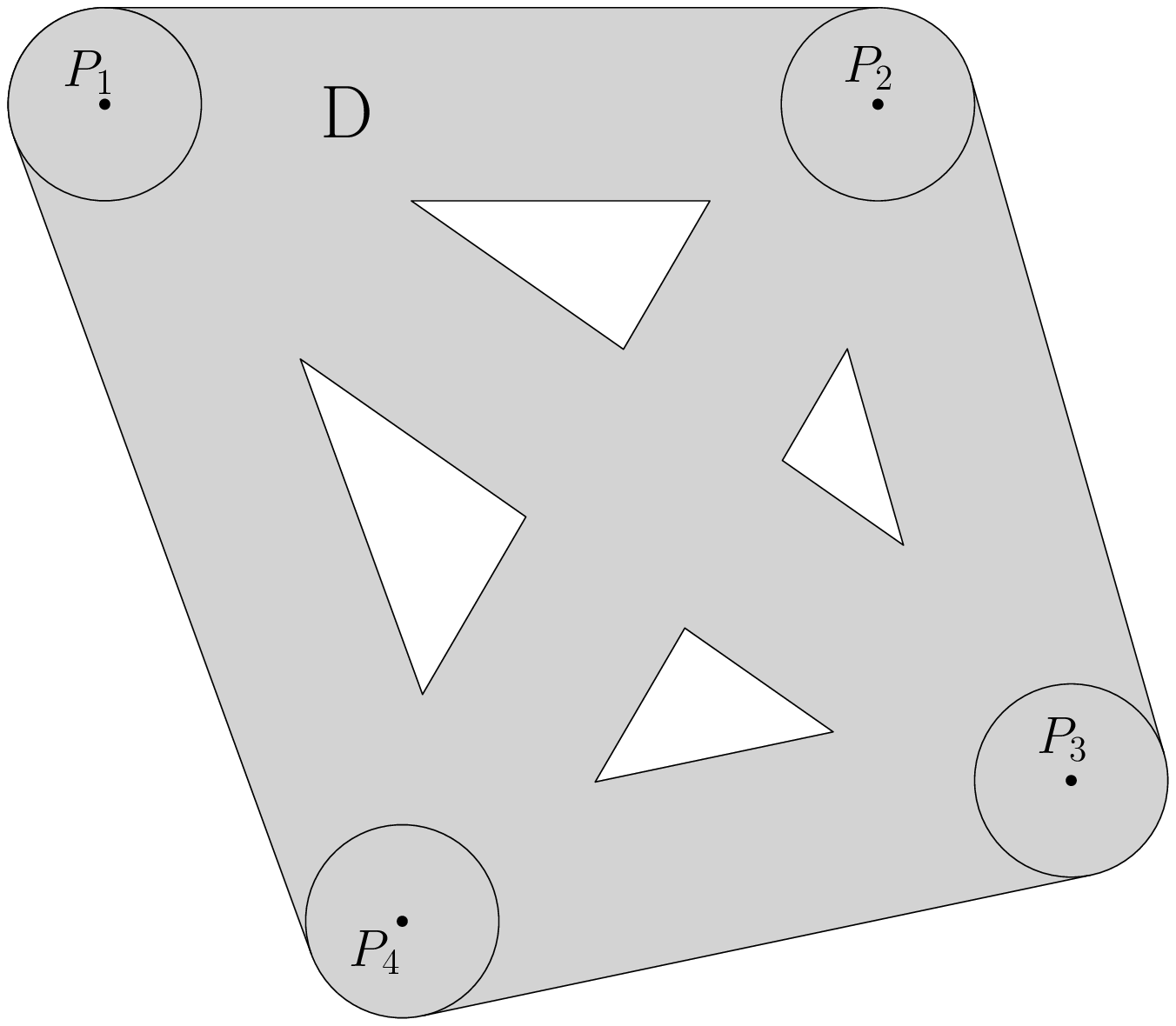}~~
	\includegraphics[width=0.48\columnwidth]{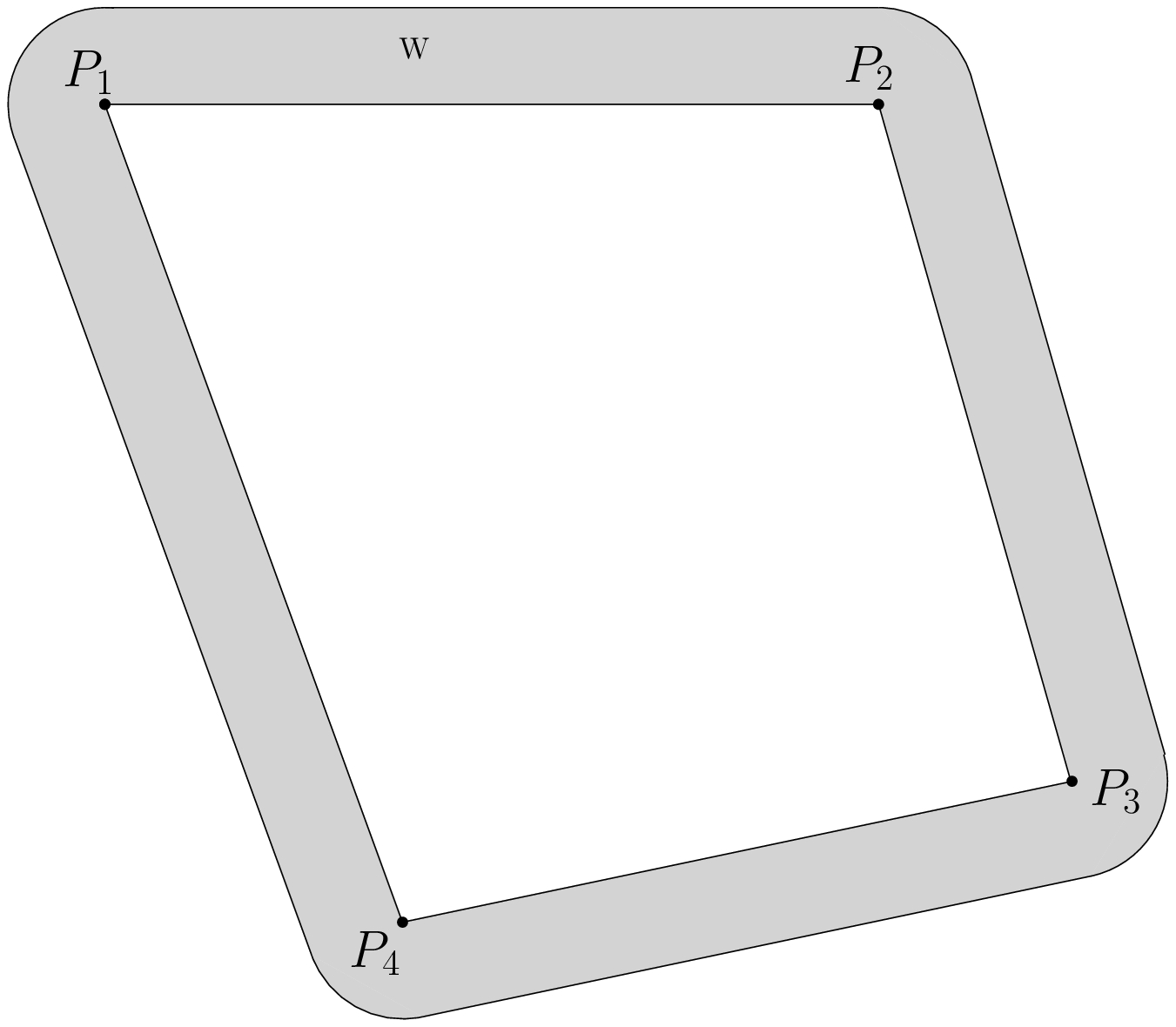}\\
	(a)\hspace*{3.4cm}(b)
	\caption{(a) Set $\widehat D$ denoting the union of all the 2-pursuer capture regions; (b) set $\widehat\creg\backslash\op{\cP}$.}
	\label{fig:union_capture_region}
\end{figure}

\begin{theorem}\label{th:capture_segment}
	If $E\in\widehat\creg\backslash\op{\cP}$, then there exists a pair of pursuers $P_i,P_j$ for which condition \eqref{eq:cond2vs1} holds.
\end{theorem}
\begin{proof}
	See Appendix.
\end{proof}

\begin{remark}\label{rem1}
	An alternative formulation of Theorem~\ref{th:capture_segment} states that if $E\in\cregm \backslash\op{\cP}$, there exists a pair of pursuers for which condition \eqref{eq:cond2vs1} holds.	
	In fact, from the definitions \eqref{eq:capture_region_multi}  and \eqref{eq:Dhat}, it can be observed that  $\widehat\creg \subseteq \cregm$. Hence, $\widehat\creg\backslash\op{\cP}  \subseteq \cregm \backslash\op{\cP}$. On the other hand, if $Q \in  \cregm \backslash\op{\cP}$, there exist two pursuers $P_i$, $P_j$, such that $Q \in \creg_{ij}$, which implies $Q \in \widehat\creg\backslash\op{\cP}$. Therefore $\widehat\creg\backslash\op{\cP} = \cregm \backslash\op{\cP}$.
\end{remark}

Notice that Theorem~\ref{th:capture_segment} guarantees the satisfaction of \eqref{eq:cond2vs1} if $E\in\partial\cP$.
Moreover, Theorem~\ref{th:capture_segment} suggests that if the pursuers play a strategy which steadily reduces the size (in some sense) of $\cP$ over time, then \eqref{eq:cond2vs1} is eventually assured. The following result guarantees the satisfaction of \eqref{eq:cond2vs1} whenever the pursuers are sufficiently close to each other.

%

\begin{theorem}\label{th:winning_strategies}
	Let us consider a pursuit strategy such that
	\begin{equation}\label{eq:winning}
		\max_{i,j=1,\ldots,p} \|P_i(\bar t)-P_j(\bar t)\|\le \sqrt{3}\rc
	\end{equation}
	at a certain time $\bar t>0$. Then, there exist a time $\hat t<\bar t$ and a pair of pursuers $P_i,P_j$ for which condition \eqref{eq:cond2vs1} holds at $\hat t$. Hence, the considered pursuit strategy is a potential switching strategy.
\end{theorem}
\begin{proof}
	See Appendix.
\end{proof}


Notice that condition \eqref{eq:winning} in Theorem~\ref{th:winning_strategies} is much simpler than \eqref{eq:cond2vs1}, since it does not depend on the evader position. Therefore, it can be effectively employed in designing families of winning switching strategies, even starting from naive pursuit strategies that would not be successful without switching. A first family consists of pursuers $P_k$ going straight towards a fixed point $M\in\rr^2$, with speed $\|v_{P_k}(t)\|\ge\eps$, for some $0 < \eps \leq 1$, and stopping once they reach $M$. It is referred to as \textit{Fixed-Point Pursuit Strategy (FPPS)}. The following result holds.
\begin{theorem}\label{th:switching_strategies}
	Any FPPS is a potential switching strategy.
\end{theorem}
\begin{proof}
	Let 
	$$
	d=\max_{k\in\{1,\ldots,p\}}\|P_k(0)-M\|.
	$$	
	Then, at time $\tau=\frac{d-\rc}{\eps}$ one has $\|P_k(\tau)-M\|\le\rc$,~ $k=1,\ldots,p$, and hence $\|P_i(\tau)-P_j(\tau)\|\le\|P_i(\tau)-M\|+\|M-P_j(\tau)\|\le2\rc$, for any pair of pursuers $P_i,P_j$. Thus, the result follows by Theorem~\ref{th:winning_strategies}.
\end{proof}

Inspired by Theorem~\ref{th:centroid}, an example of FPPS is the \emph{centroid} strategy, in which all pursuers (or a subset of them, including the evader in their convex hull) move with the same velocity towards the center of the minimum-radius circle which includes them.

Another simple pursuit strategy is the so-called \emph{pure pursuit}. In this strategy, all pursuers always point towards the evader at maximum velocity, i.e.
\begin{equation}
	v_{P_k}(t)= \frac{E(t)-P_k(t)}{\|E(t)-P_k(t)\|},\quad k=1,\ldots,p.
	\label{eq:vPpp}
\end{equation}
It is well known that pure pursuit does not guarantee capture of the evader. However, it becomes a winning strategy if switching to \tg is adopted.
\begin{theorem}\label{th:PP_pot_optimal}
	The pure pursuit strategy is a potential switching strategy.
\end{theorem}
\begin{proof}
	Let $\delta_k(t)=\|E(t)-P_k(t)\|$. Then, by using \eqref{eq:vPpp},
	\begin{align*}
		\ds\frac{1}{2} \ds\frac{d \delta^2_k(t)}{dt} &= (E(t)-P_k(t))'(\dot{E}(t)-\dot{P}_k(t))
		=(E(t)-P_k(t))' \dot{E}(t) - \delta_k(t).
	\end{align*}
	Being $\|\dot{E}(t)\|=\|v_E(t)\|\leq 1$, one has $ \frac{d \delta^2_k(t)}{dt} \leq 0$ and in particular $ \frac{d \delta^2_k(t)}{dt} = 0$ if and only if $v_E(t)=v_{P_k}(t)$ given by \eqref{eq:vPpp}. Hence, $\lim_{t \rightarrow +\infty} \delta_k(t) = \bar\delta_k \geq 0$ and $\lim_{t \rightarrow +\infty} \dot\delta_k(t) = 0$, which in turns leads to $\lim_{t \rightarrow +\infty} v_{P_k}(t) = \lim_{t \rightarrow +\infty} v_E(t)$, $\forall k=1,\dots,p$. 
	Since, according to \eqref{eq:vPpp}, all pursuers point towards the evader, this means that all pursuers and the evader tend to be asymptotically aligned, with the evader lying outside the segment containing the pursuers. Therefore, $\lim_{t \rightarrow +\infty} E(t) \notin \cP$. By continuity of the agents' trajectories, $E(t)$ must have crossed $\partial\cP$ at some time instant and hence the conclusion follows by Theorem~\ref{th:capture_segment}.
\end{proof}

Theorems~\ref{th:switching_strategies} and \ref{th:PP_pot_optimal} are just examples of winning pursuit strategies based on switching; it is apparent that many other switching strategies can be devised by exploiting the geometric conditions in Theorems~\ref{th:capture_segment} and \ref{th:winning_strategies}. 

\section{Numerical simulations}\label{sec:simulations}

In order to show the benefits of using a switching strategy, several numerical simulations are reported in this section. Three pursuit strategies are considered: pure-pursuit (PPS), centroid-based (CS), and the strategy reported in \cite{zhou2016}, referred to as Voronoi-based strategy (VS). 
Notice that in \cite{zhou2016} it was proven that VS is a winning strategy.
The prefix S- is used to denote a given pursuit strategy when pursuers switch to \tg as soon as condition~\eqref{eq:cond2vs1} holds. For instance, S-PPS denotes the switching pure-pursuit strategy.

In all simulations, the maximum speed of the players is set to $1$ and the capture radius is $r=1$. In the figure illustrating the agents' trajectories, pursuers are drawn in blue and the evader in red. Initial conditions are marked with a square, while the final positions of the agents are denoted by a dot. A dashed circle with radius $r$ is drawn around the capture position of the evader. 

\begin{example}\label{ex:1}
	Let us consider a four-pursuer one-evader game, where pursuers are initially located at the vertices of a square with side length $20$ centered at the origin, i.e., $P_1(0)=[-10,-10]'$, $P_2(0)=[-10,10]'$, $P_3(0)=[10,-10]'$, $P_4(0)=[10,10]'$. The evader starts from $E(0)=[0,5]'$, and it goes upwards with maximum speed during each game.
	
	If the pursuers move according to PPS and CS, the evader can escape, and so such pursuit strategies are not winning in this scenario. On the contrary, if pursuers switch when \eqref{eq:cond2vs1} occur, the evader is captured at time $t=22.66$ and $t=33.91$ for S-PPS and S-CS, respectively. The trajectories traveled by the players in such cases are depicted in Fig.~\ref{fig:Example1_a} and \ref{fig:Example1_b}.
	
	When pursuers play the original version of the Voronoi-based strategy, they are able to capture the evader at time $t=20.02$, as depicted in Fig.~\ref{fig:Example1_c}. When playing the corresponding switching version, capture time is reduced to $t=16.61$. So, in this example, the capture time needed by VS is about $20\%$ greater than that needed by S-VS. The agents' trajectories are shown in Fig.~\ref{fig:Example1_d}, where it can be observed that in this case the capture point coincides with the centroid of all the four pursuers, after condition \eqref{eq:cond2vs1} has occurred.

	\begin{figure}[ht]
	\centering
	\includegraphics[width=0.6\linewidth]{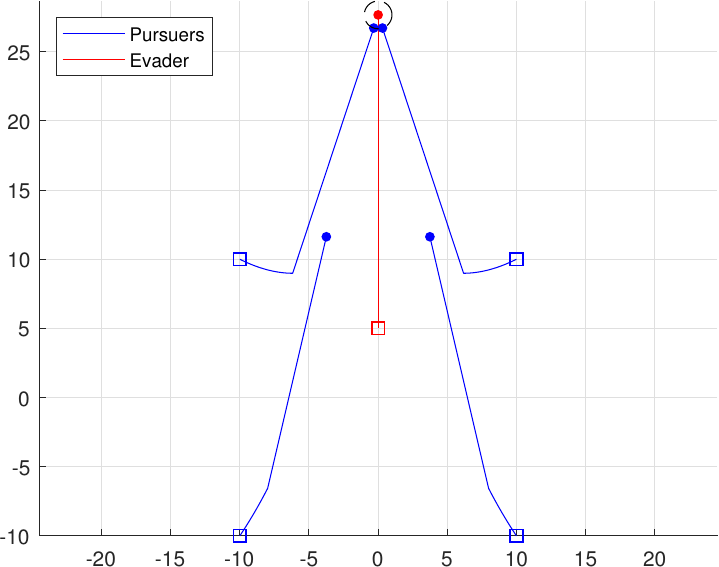}
	\caption{Example~\ref{ex:1}. The evader goes upwards while pursuers play S-PPS. Capture occurs at  $t=22.66$.}
	\label{fig:Example1_a}
\end{figure}
\begin{figure}[ht]
	\centering
	\includegraphics[width=0.6\linewidth]{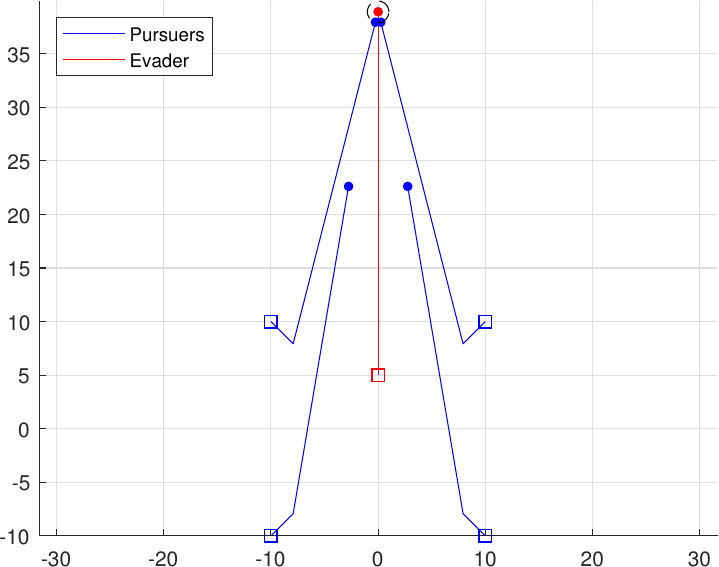}
	\caption{Example~\ref{ex:1}. The evader goes upwards while pursuers play S-CS. Capture occurs at  $t=33.91$.}
	\label{fig:Example1_b}
\end{figure}
	
	\begin{figure}[th]
		\centering
		\includegraphics[width=0.6\linewidth]{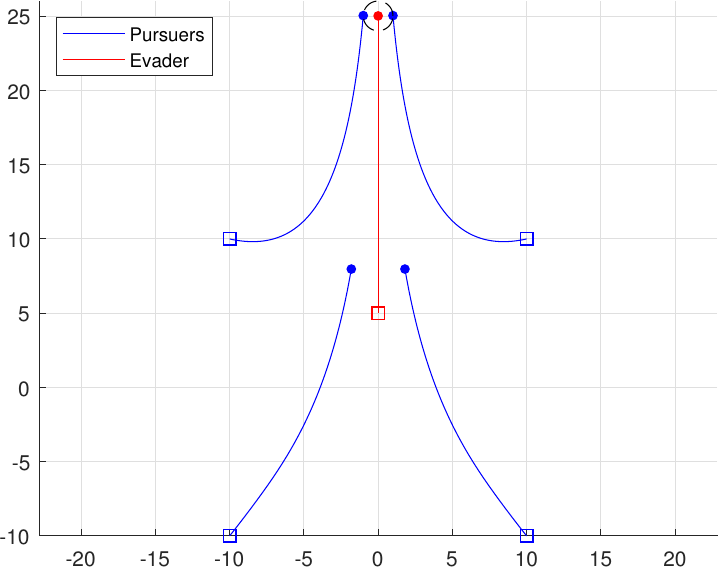}
		\caption{Example~\ref{ex:1}. The evader goes upwards while pursuers play VS. Capture occurs at  $t=20.02$.}
		\label{fig:Example1_c}
	\end{figure}
	\begin{figure}[th]
		\centering
		\includegraphics[width=0.6\linewidth]{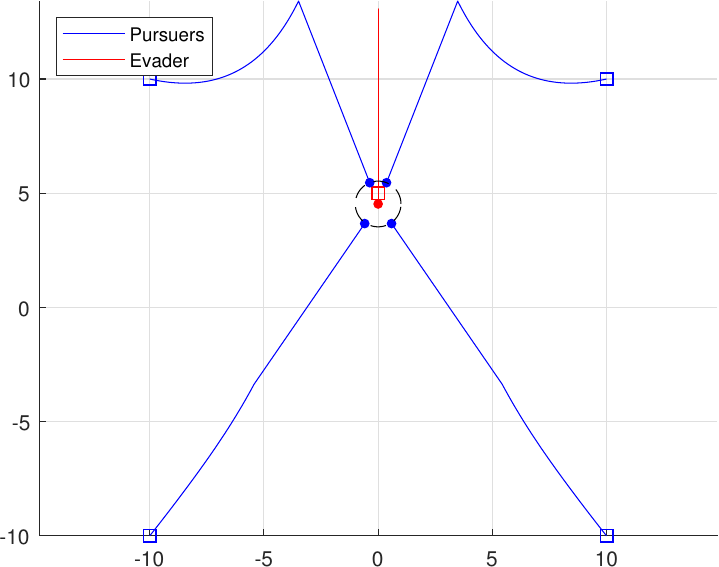}
		\caption{Example~\ref{ex:1}. The evader goes upwards while pursuers play S-VS. Capture occurs at  $t=16.61$.}
		\label{fig:Example1_d}
	\end{figure}
\end{example}

\begin{example}\label{ex:2}
	The aim of this example is to evaluate the benefit of using S-VS in place of VS. To this purpose, a simulation campaign with different number of pursuers located in different places is performed. In all simulations, the evader position is initially set to the origin. Pursuers start randomly inside a circle of radius $R_p$ centered at the origin. Only initial positions such that the evader belongs to the convex hull of the pursuers are considered.
	
	For a generic game, let us denote by $T_{_{VS}}$ and $T_{_{S\text{-}VS}}$ the capture time when pursuers play VS and S-VS, respectively. The convenience of adopting the switching strategy is measured through the ratio of the mentioned capture times, that is
	$
	\rho=\frac{T_{_{VS}}}{T_{_{S\text{-}VS}}}.
	$
	The evader strategy consists in pointing towards the farthest vertex of its own Voronoi cell at each time. 
	
	The number $p$ of pursuer ranges from $3$ to $10$, while the radius of the circle of the initial pursers' position varies as $R_p\in\{10,30,50\}$. For all possible combinations of $p$ and $R_p$, $100$ games are played. The corresponding boxplots of $\rho$ are reported in Fig.~\ref{fig:Example2_boxplots}. To increase readability, the upper limit of the $\rho$-axis is set to $2$. However, it is worthwhile stressing that for several games much larger values have been observed. For instance, for $p=3$ and $R_p=50$ the maximum ratio is $\rho=12.6$. Notice that in games with $\rho=1$ the switching condition \eqref{eq:cond2vs1} never occurs, so that $T_{_{VS}}=T_{_{S\text{-}VS}}$. On the whole, the simulation results demonstrate that playing the switching strategy leads to significantly smaller capture times than adopting the original Voronoi-based strategy.

\clearpage
	
	\begin{figure}[t]
		\centering
		\includegraphics[width=.8\linewidth]{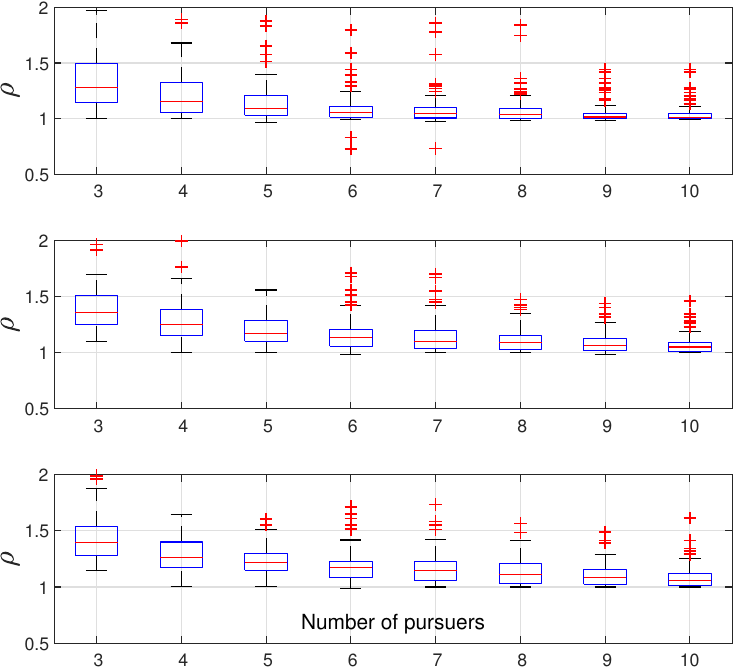}
		\caption{Example~\ref{ex:2}. Boxplot of the performance index $\rho$ for $R_p=10$ (top), $R_p=30$ (middle) and $R_p=50$ (bottom).}
		\label{fig:Example2_boxplots}
	\end{figure}
\end{example}

\section{Conclusions}\label{sec:conclusions}

A new family of switching pursuit strategies for a multi-pursuer single-evader game has been introduced. It is based on the key observation that the game can be reduced to a two-pursuer single-evader setting, for which a minimum-time solution is available. Conditions for transitioning between these two games have been derived, enabling the definition of new winning strategies and enhancing the performance of existing ones. Future work will focus on developing pursuit strategies that achieve the switching condition while minimizing the overall capture time or other relevant performance indexes. Moreover, the approach will be extended to handle multi-pursuer multi-evader games. The use of switching pursuit strategies in scenarios involving a superior evader will also be considered.

\clearpage
\section*{Appendix}

\subsection*{Proof of Theorem \ref{th:capture_segment}}

In order to prove Theorem \ref{th:capture_segment}, the following lemmas are introduced.
\begin{lemma}\label{lem:greater_d}
	Let $P_1,P_2$ and $E$ be such that $E\in\creg_{12}\backslash(\cC(P_1,\rc)\cup\cC(P_2,\rc))$, and set $P_3=P_2+\alpha (P_2-P_1)$, $\alpha>0$. Then, $T_{13}>T_{12}$.
\end{lemma}
\begin{proof}
	W.l.o.g., let us refer to the reference frame introduced in Section~\ref{subsec:recall_2P1E}, see Fig.~\ref{fig:P3_far}. 
	Hence, from $E=[x,y]' \in\creg_{12}\backslash(\cC(P_1,\rc)\cup\cC(P_2,\rc))$, one has $-d <x <d$ and $(x \pm d)^2+y^2-r^2 >0$.
	Moreover, since $E\in\creg_{12}\backslash(\cC(P_1,\rc)\cup\cC(P_2,\rc))$, one has also $E\in\creg_{13}\backslash(\cC(P_1,\rc)\cup\cC(P_3,\rc))$.
	Therefore, the capture time $T_{13}$ of pursuers $P_1$ and $P_3$ can be computed by using \eqref{eq:function_fp}, in which $d$ and $x$ are replaced by $d+\alpha d$ and $x -\alpha d$, respectively. 
	One gets
	\begin{equation}\label{eq:T13}
		T_{13}=\frac{\kappa_\alpha \rc+|y|\sqrt{\kappa_\alpha^2-4(x-\alpha d)^2(\rc^2-y^2)}}{2(\rc^2-y^2)}
	\end{equation}
	in which $\kappa_\alpha=\kappa+2 \alpha d (d+x)$.
	In order to prove $T_{13}>T_{12}$, it is sufficient to show that $T_{13}$ in \eqref{eq:T13} is a strictly increasing function of $\alpha$, for $\alpha>0$.
	First, one has
	$$
	\frac{d \kappa_\alpha}{d \alpha}=2d(d+x)>0.
	$$
	Then, it remains to show that $\kappa_\alpha^2-4(x-\alpha d)^2(\rc^2-y^2)$ is strictly increasing in $\alpha$. Through some straightforward manipulations, one obtains
	$$
	\begin{array}{l}
		\ds\frac{d}{d \alpha} \left\{\kappa_\alpha^2-4(x-\alpha d)^2(\rc^2-y^2) \right \}
		= 4 \left\{ d(d-x)+2d^2 \alpha \right\} \left\{ (x+d)^2+y^2-r^2 \right\}
	\end{array}
	$$
	which is positive for all $\alpha>0$.
\end{proof}

\begin{lemma}\label{lem:E_belongs_D}
	Consider $P_1$, $P_2$ and $E$ such that $E \in \cD_{12} \backslash (\cC(P_1,\rc) \cup \cC(P_2,\rc) )$. Let $\cL_1$ and $\cL_2$ be the lines containing the segments $\overline{P_1 E}$ and $\overline{P_2 E}$, respectively. Then, let $Q_1$ and $Q_2$ be the intersections of $\cL_1$ and $\cL_2$ with the circumference $\partial\cC(H_{12}^+,T_{12}+\rc)$, opposite to $P_1$ and $P_2$, respectively. Define the region $\cW_{12}$ as the portion of the circle $\cC(H_{12}^+,T_{12}+\rc)$ bounded by the segments $\overline{Q_1 E}$, $\overline{Q_2 E}$ and the arc of circumference $\wideparen{Q_2 Q_1}$. Then, if a pursuer $P_3 \in \cW_{12}$, then either $E \in \cD_{13}$ or $E \in \cD_{23}$.
\end{lemma}	
\begin{proof}
	Let us consider the region $\cW_{12}$ depicted in Fig.~\ref{fig:E_belongs_D}. Let $S=\partial\cC(E,\rc)\cap \partial\cC(H_{12}^+,T_{12}+\rc)$ be the point of tangency between the two circumferences. Notice that the segment $\overline{ES}$ divides $\cW_{12}$ in two disjoint regions. Thus, either $\overline{P_1 P_3}$ or $\overline{P_2 P_3}$ will intersect $\overline{ES}$.
	Let us assume that $\overline{P_2 P_3}\cap\overline{ES}\neq\emptyset$. Then, the minimum distance between $E$ and $\overline{P_2 P_3}$ is less than $\rc$, and hence $E\in\creg_{23}$. A similar argument can be adopted if $\overline{P_1 P_3}\cap\overline{ES}\neq\emptyset$.
	\begin{figure}[t]
		\centering
		\includegraphics[width=0.6\columnwidth]{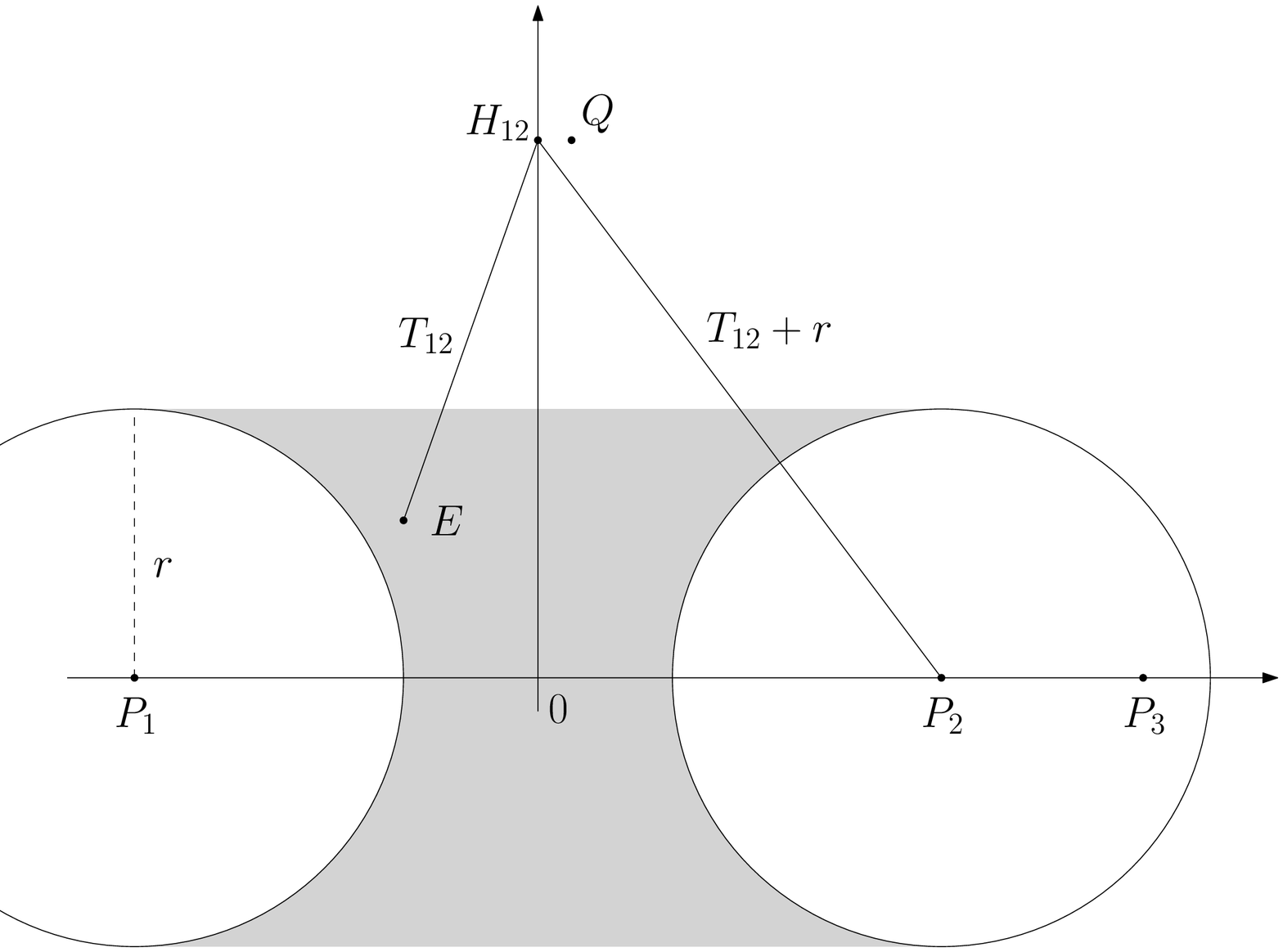}
		\caption{Sketch of the proof of Lemma~\ref{lem:greater_d}. The gray region denotes $\creg_{12}\backslash(\cC(P_1,\rc)\cup\cC(P_2,\rc))$.}
		\label{fig:P3_far}
	\end{figure}	
	\begin{figure}[t]
		\centering
		\includegraphics[width=0.6\columnwidth]{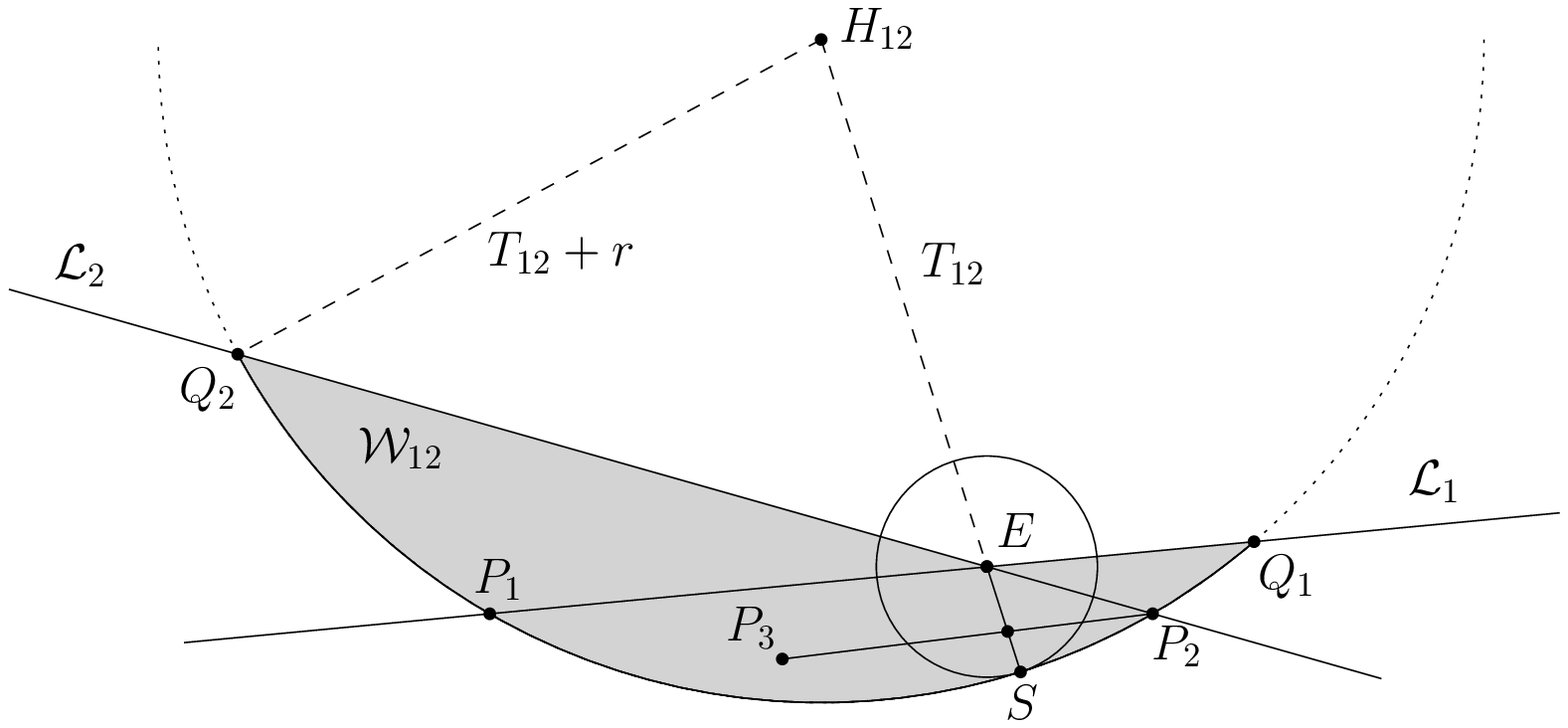}
		\caption{Sketch of the proof of Lemma~\ref{lem:E_belongs_D}.}
		\label{fig:E_belongs_D}
	\end{figure}
\end{proof}

Now, we can proceed to prove Theorem \ref{th:capture_segment}.
Let $P_i,P_j$ be the pair of pursuers such that $E \in \cD_{ij}$ and $T_{ij} \leq T_{hl}$, for all pairs $h,l \in \{1,...,p\}$ for which $E \in \cD_{hl}$.
We show that for this pursuer pair, condition \eqref{eq:cond2vs1b}, and hence \eqref{eq:cond2vs1}, holds.
Condition \eqref{eq:cond2vs1b} states that $\|P_k-H_{ij}\|\ge\|P_i-H_{ij}\|=T_{ij}+\rc$, ~$k=1,\ldots,p$.
By contradiction, assume that there exists a pursuer $P_z$ which satisfies $\|P_z-H_{ij}\|<T_{ij}+\rc$, i.e., $P_z\in\op{\cC(H_{ij},T_{ij}+\rc)}$.
Let us define $\cL_{iE}$ and $\cL_{jE}$ as the lines passing through $P_i,E$ and $P_j,E$, respectively.
Since $P_z\in\cP$ and by assumption $E\notin\op{\cP}$, then $P_z$ must belong to the region $\cW_{ij}$, shown in Fig.~\ref{fig:E_segment}. 
By Lemma~\ref{lem:E_belongs_D}, either $E\in\creg_{iz}$ or $E\in\creg_{jz}$. W.l.o.g., assume $E\in\creg_{jz}$.
Let $\cL_{jz}$ be the line crossing $P_j$ and $P_z$. It holds $\cL_{jz}\cap\cC(H_{ij},T_{ij}+\rc)=\{P_j,P_q\}$, where $P_q$ denotes a fictitious pursuer.
Then, by Lemma~\ref{lem:greater_d} one has $T_{jz}<T_{jq}=T_{ij}$.
Since by assumption $T_{ij}\le T_{hl}$ for all $h,l\in\{1,\ldots,p\}$, a contradiction occurs.
\unskip\nobreak\hfill $\blacksquare$

\bigskip

\begin{figure}[t]
	\centering
	\includegraphics[width=0.6\columnwidth]{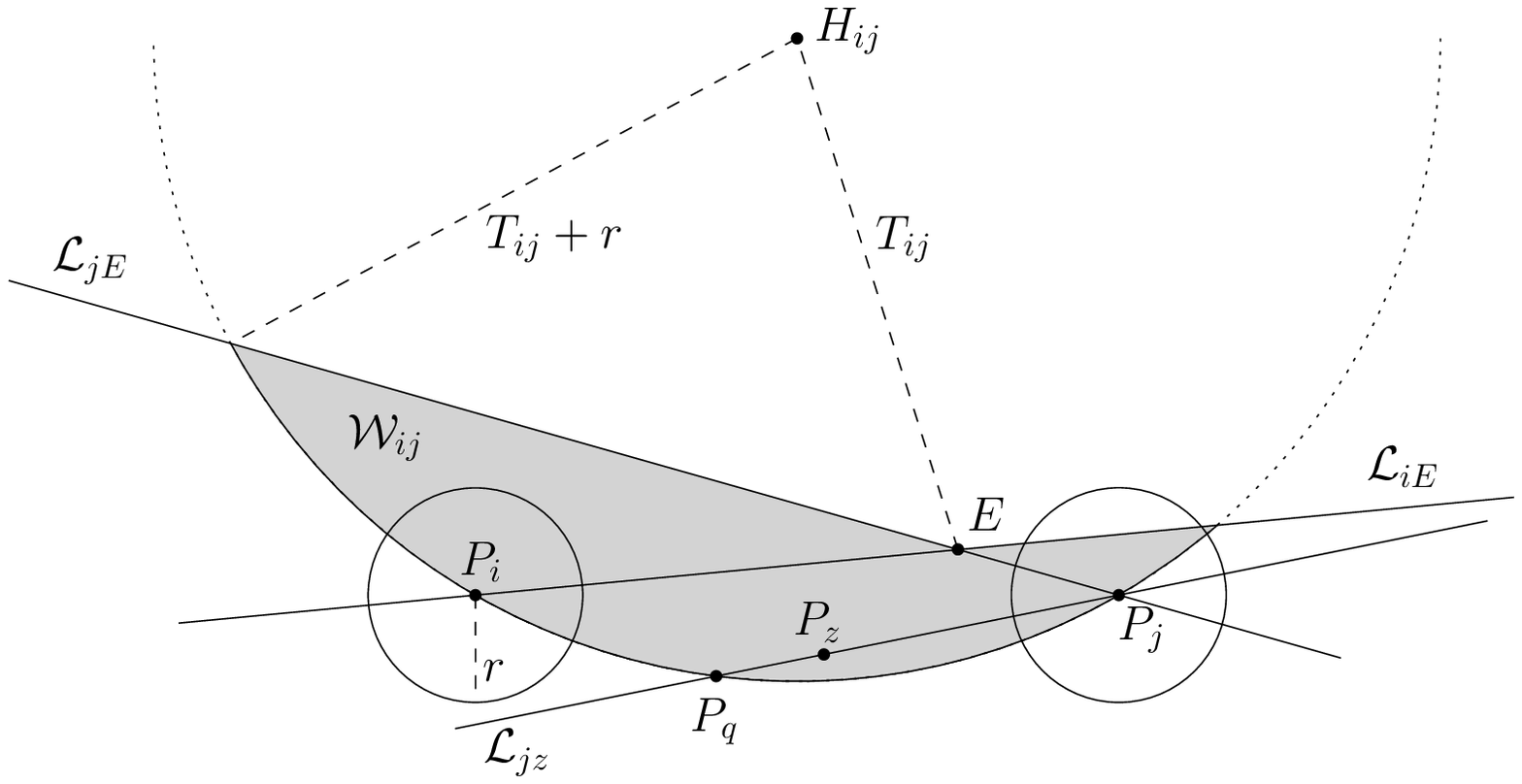}
	\caption{Sketch of the proof of Theorem~\ref{th:capture_segment}. The gray area $\cW_{ij}$ denotes the admissible region of $P_z$.}
	\label{fig:E_segment}
\end{figure}

\subsection*{Proof of Theorem \ref{th:winning_strategies}}
According to Theorem~\ref{th:capture_segment} and Remark~\ref{rem1}, if $E(0)\in\widehat\cregm(0)\backslash\op{\cP(0)}$, condition \eqref{eq:cond2vs1} already holds at time $t=0$. So, let us analyze the case $E(0)\in\op{\cP(0)}$.
Let us first show that \eqref{eq:winning} implies 
\begin{equation}\label{eq:cover}
	\cP(\bar t)\subset\bigcup_{i=1}^p \cC(P_i(\bar t),\rc).
\end{equation} 
If $p=3$ and the distance between each pursuer pair is exactly equal to $\sqrt{3}r$, then $\cP$ is an equilateral triangle and \eqref{eq:cover} follows from simple geometric arguments. Clearly, if the pursuers are closer to each other, \eqref{eq:cover} still holds. In the generic case of $p$ pursuers, $\cP$ can be partitioned in triangles and the previous reasoning can be repeated for each triangle. Therefore, \eqref{eq:cover} is satisfied for a generic polytope $\cP$.
This means that $E(\bar t)\notin\cP(\bar t)$, otherwise capture had already occurred. Since $E(0)\in\op{\cP(0)}$, by continuity of the agents' trajectories there exists $\hat t<\bar t$ such that $E(\hat t)\in\partial\cP(\hat t)$ and then, by Theorem~\ref{th:capture_segment}, condition \eqref{eq:cond2vs1} holds at time $\hat t$.
\unskip\nobreak\hfill $\blacksquare$


\bibliographystyle{IEEEtran}
\bibliography{lion_and_man}

\end{document}